\documentclass[11pt]{article}

\usepackage[utf8]{inputenc}
\usepackage{newunicodechar}
\usepackage[T1]{fontenc}
\usepackage{lmodern}

\usepackage{amsmath,amssymb,amsfonts,mathtools,amsthm}

\usepackage[linesnumbered,vlined]{algorithm2e}
\DontPrintSemicolon

\newcommand{\concept}[1]{\textbf{#1}}

\newcommand{\etal}{\textit{et al.}}

\newcommand{\Tuple}[1]{\left\langle{#1}\right\rangle}

\newtheorem{theorem}{Theorem}[section]
\newtheorem{lemma}[theorem]{Lemma}
\newtheorem{corollary}[theorem]{Corollary}

\newunicodechar{¬}{\neg}   % ¬
\newunicodechar{∧}{\wedge} % ∧
\newunicodechar{⊕}{\oplus} % ⊕
\newunicodechar{∨}{\vee}   % ∨
\newunicodechar{→}{\rightarrow}  % →
\newunicodechar{⇒}{\Rightarrow}  % ⇒
\newunicodechar{←}{\leftarrow}   % ←
\newunicodechar{↔}{\leftrightarrow}   % ↔
\newunicodechar{⊢}{\vdash}   % ⊢

\newunicodechar{↦}{\mapsto}  % ↦
\newunicodechar{↪}{\hookrightarrow}  % ↪
\newunicodechar{↣}{\rightarrowtail}  % ↣

\newunicodechar{∀}{\forall} % ∀
\newunicodechar{∃}{\exists} % ∃

\newunicodechar{∩}{\cap}    % ∩
\newunicodechar{∪}{\cup}    % ∪
\newunicodechar{∖}{\setminus}    % ∖
\newunicodechar{×}{\times}    % ×
\newunicodechar{∈}{\in}     % ∈
\newunicodechar{∉}{\not\in}     % ∉
\newunicodechar{∋}{\ni}     % ∋
\newunicodechar{⊆}{\subseteq}    % ⊆
\newunicodechar{⊇}{\supseteq}    % ⊇
\newunicodechar{⊑}{\sqsubseteq}    % ⊑
\newunicodechar{⊒}{\sqsupseteq}    % ⊑
\newunicodechar{≾}{\lesssim}
\newunicodechar{≿}{\grtsim}

\newunicodechar{∘}{\circ}    % ∘
\newunicodechar{⋅}{\cdot}    % ⋅

\newunicodechar{∅}{\emptyset}  % ∅
\DeclareUnicodeCharacter{2115}{\mathbb{N}} % ℕ
\DeclareUnicodeCharacter{2124}{\mathbb{Z}} % ℤ
\DeclareUnicodeCharacter{211A}{\mathbb{Q}} % ℚ
\DeclareUnicodeCharacter{211D}{\mathbb{R}} % ℝ
\DeclareUnicodeCharacter{2102}{\mathbb{C}} % ℂ
\newunicodechar{ℵ}{\aleph}     % ℵ
\newunicodechar{ℶ}{\beth}      % ℶ

\DeclareUnicodeCharacter{2131}{\mathcal{F}} % ℱ

\newunicodechar{∑}{\sum}  % ∑
\newunicodechar{∏}{\prod} % ∏
\newunicodechar{∫}{\int}  % ∫
\newunicodechar{±}{\pm}   % ±

\newunicodechar{≠}{\neq}    % ≠
\newunicodechar{≡}{\equiv}  % ≡
\newunicodechar{≈}{\approx} % ≈
\newunicodechar{≃}{\simeq}  % ≃
\newunicodechar{≅}{\cong}   % ≅
\newunicodechar{≤}{\le}     % ≤
\newunicodechar{≥}{\ge}     % ≥
\newunicodechar{∣}{\divides}% ∣

\newunicodechar{⊤}{\top}    % ⊤
\newunicodechar{⊥}{\bot}    % ⊥
\newunicodechar{∞}{\infty}  % ∞
\newunicodechar{§}{\S}      % §
\newunicodechar{◇}{\Diamond}      % ◇

\newunicodechar{Γ}{\Gamma}  % Γ
\newunicodechar{Δ}{\Delta}  % Δ
\newunicodechar{Θ}{\Theta}  % Θ
\newunicodechar{Λ}{\Lambda} % Λ
\newunicodechar{Ξ}{\Xi}     % Ξ
\newunicodechar{Π}{\Pi}     % Π
\newunicodechar{Σ}{\Sigma}  % Σ
\newunicodechar{Φ}{\Phi}    % Φ
\newunicodechar{Ω}{\Omega}  % Ω
\newunicodechar{α}{\alpha}  % α
\newunicodechar{β}{\beta}   % β
\newunicodechar{γ}{\gamma}  % γ
\newunicodechar{δ}{\delta}  % δ
\newunicodechar{ε}{\epsilon} % ε
\newunicodechar{ζ}{\zeta}   % ζ
\newunicodechar{η}{\eta}    % η
\newunicodechar{κ}{\kappa}  % κ
\newunicodechar{λ}{\lambda} % λ
\newunicodechar{μ}{\mu}     % μ
\newunicodechar{ξ}{\xi}     % ξ
\newunicodechar{π}{\pi}     % π
\newunicodechar{ρ}{\rho}    % ρ
\newunicodechar{σ}{\sigma}  % σ
\newunicodechar{τ}{\tau}    % τ
\newunicodechar{φ}{\phi}    % φ
\newunicodechar{χ}{\chi}    % χ
\newunicodechar{ψ}{\psi}    % ψ
\newunicodechar{ω}{\omega}  % ω

\DeclarePairedDelimiter{\card}{\lvert}{\rvert}

\newcommand{\Set}[1]{\left\{{#1}\right\}}
\newcommand{\SetWhere}[2]{\left\{{#1}\;\middle\vert\;{#2}\right\}}

\SetKwFor{Procedure}{procedure}{}{end procedure}

\newcommand{\newData}[2]{\newcommand{#1}{\DataSty{#2}\xspace}}

\newData{\Leader}{leader}
\newData{\Frozen}{frozen}

\title{Stably computable relations and predicates}

\author{James Aspnes\\Yale University}

\date{}

\begin{document}

\maketitle

\begin{abstract}
    A population protocol stably computes a relation $R(x,y)$ if 
    its output always stabilizes and $R(x,y)$ holds 
    if and only if $y$ is a possible output for input $x$.
    Alternatively, a population protocol computes a predicate
    $R(\Tuple{x,y})$ on pairs
    $\Tuple{x,y}$ if its output stabilizes on the truth value of the predicate when
    given $\Tuple{x,y}$ as input.

    We consider how stably computing $R(x,y)$ and $R(\Tuple{x,y})$
    relate to each other. We show that for population protocols
    running on a complete interaction graph with $n≥2$, if
    $R(\Tuple{x,y})$ is a stably computable predicate such that
    $R(x,y)$ holds for at least one $y$ for each $x$, then $R(x,y)$ is a stably
    computable relation. In contrast, the converse is not necessarily true
    unless $R(x,y)$ holds for exactly one $y$ for each $x$.
\end{abstract}

\section{Introduction}

A \concept{population protocol}~\cite{AngluinADFP2006} models a
distributed computation by a collection of $n$ finite-state agents that
interact directly in pairs determined by the edges of a directed
\concept{interaction graph} whose vertices are the agents.
Each agent has a state from some finite
\concept{state set} $Q$.
The states of all agents are described by a \concept{configuration}
$C∈Q^n$, and
this configuration is updated 
at each step of an execution when an \concept{adversary scheduler}
chooses two agents to interact, updating their states according to a
joint \concept{transition function} $δ:Q×Q→Q×Q$ while leaving the
states of all other agents unchanged. The adversary is constrained by
a \concept{global fairness condition} that says if $C→C'$ is a
possible transition and $C$ occurs infinitely often in an execution, so does $C'$.

The \concept{input} $x$ to a population protocol is an element of $X^n$ for some
\concept{input alphabet} $X$, and is translated to
initial states of the agents by an \concept{input map} $I:X→Q$.
A similar \concept{output map} $O:Q→Y$ translates states to
an output alphabet $Y$, yielding an output $y∈Y^n$.
A protocol is \concept{output-stable} if it eventually converges to a
fixed output in all fair executions. 

A population protocol \concept{stably computes} the
relation $R(x,y)$ that holds whenever $x$ is an input and $y$ is a
possible stable output arising in some execution starting with that
input. A relation is \concept{stably computable} if there is some output-stable
protocol that stably computes it. Because an output-stable protocol
must eventually stabilize on some output, any stably computable
relation $R(x,y)$ is \concept{total} in the sense that for any input $x$ there
is at least one output $y$ such that $R(x,y)$ holds.
Because population protocols are in general non-deterministic, this
output is not necessarily unique. In the case where there is
exactly one such $y$, $R(x,y)$ is \concept{single-valued}

Most work on population protocols considers stable computation of
predicates, where every agent outputs the same value $0$ or $1$ when
the computation finishes. For a complete interaction graph, the class of predicates that
can be stably computed has been shown~\cite{AngluinAER2007} to be precisely the semilinear
predicates, those for which the counts of agents in each state in the
set of positive inputs 
form a finite union of sets of the form
$\SetWhere{\vec{b}+∑_{i=1}^k a_i\vec{x_i}}{a_i∈ℕ}$.
But the more general case of stably computable relations is
less well understood.

\subsection{Our results}

Given two vectors of equal length, let their \concept{zip}
$\Tuple{x,y}$ be a single vector of pairs defined by $\Tuple{x,y}_i =
\Tuple{x_i,y_i}$.  There is a natural correspondence between relations
$R(x,y)$ and predicates $R(\Tuple{x,y}) = R(x,y)$, where we abuse
notation by using the same name for both the relation and its
corresponding predicate.

We consider how a stable computation of $R(x,y)$, considered as
a relation between inputs $x$ and outputs $y$, relates to stable
computation of $R(\Tuple{x,y})$, considered as a predicate on inputs
$\Tuple{x,y}$.

For the case of a complete
interaction graph with sufficiently many agents, we will show that:
\begin{enumerate}
    \item If $R(\Tuple{x,y})$ is a stably computable
        predicate such that $R(x,y)$ is total,
        then $R(x,y)$ is a stably computable relation.
        (Theorem~\ref{theorem-predicate-implies-relation}.)
    \item The converse is not true: there are stably computable
        relations $R(x,y)$ for which the corresponding predicate
        $R(\Tuple{x,y})$ is not stably computable. In particular, the
        reachability relation $C \stackrel{*}{→} C'$ for any population
        protocol is stably computable as a relation, but there are
        protocols for which the corresponding reachability predicate
        $\stackrel{*}{→}(\Tuple{C,C'})$ is not.
        (Corollary~\ref{corollary-stably-computable-relation-not-predicate}.)
    \item However, in the special case where $R(x,y)$ is
        \concept{single-valued}, meaning that for each $x$ there is
        exactly one $y$ satisfying $R(x,y)$,
        then $R(\Tuple{x,y})$ is a stably computable
        predicate. (Theorem~\ref{theorem-single-valued}.)
\end{enumerate}

The assumption of sufficiently many agents avoids having to deal with
some annoying small cases. Formally, we write that a protocol
stably computes a predicate $P(x)$ for $n≥n_0$ if it stabilizes
to the correct output in all fair executions with $n≥n_0$; and that it stably computes a
relation $R(x,y)$ for $n≥n_0$ if, for any starting input $x$ on
$n≥n_0$ agents, it always stabilizes in all fair executions with
$n≥n_0$, and $R(x,y)$ holds if and only if $y$ is a possible output in
at least one such execution. We show in Lemma~\ref{lemma-small-cases}
that any relation that is stably computable for all $n≥n_0$ for some
constant $n_0$ is stably computable for all $n≥2$, and also for all $n≥1$
provided it is single-valued when $n=1$.

\section{Stably computable predicates give stably computable
relations}

To show that a stably computable predicate $R(\Tuple{x,y})$
corresponds to a stably computable
relation $R(x,y)$, we construct a
protocol that on input $x$ wanders nondeterministically through all
possible values of $y$, stopping when $R(\Tuple{x,y})$ holds. The test
for $R(\Tuple{x,y})$ is done using an embedded protocol implementing
the algorithm of Angluin~\etal~\cite{AngluinACFJP2005} for computing a
semilinear predicate with stabilizing inputs; this has the property of
eventually stabilizing on the value of $R(\Tuple{x,y})$ at all agents
if $\Tuple{x,y}$ remains unchanged for long enough. Once all agents
agree that $R(\Tuple{x,y})$ holds, no further changes are made to $y$.
Managing changes to $y$ and execution of the $R(\Tuple{x,y})$ protocol
is handled by a leader elected using the usual algorithm based on the
transition $LL→LF$~\cite{AngluinADFP2006}.

We assume that the interaction graph is complete and that $n≥3$. The
restriction on $n$ will be revisited in
Section~\ref{section-small-cases}.

Pseudocode for the transition function is given in
Algorithm~\ref{alg-predicate-implies-relation}.
Each agent stores four fields: its input $x$, which does not change
during the computation; a candidate output $y$, which starts in some
default value and is updated nondeterministically; a \Leader bit
that is initially $1$ for all agents; and a state $q$ for the embedded
$R(\Tuple{x,y})$ protocol. We write $δ_R$ for the transition function
of the embedded protocol and assume that the output of the embedded
protocol can be read from a component $q.\Frozen$ that eventually 
stabilizes to $1$ when $R(\Tuple{x,y})$ holds. This is used to
prevent any further updates to $y$ once an acceptable $y$ is generated.

\begin{algorithm}
    \Procedure{$δ(i,r)$}{
        \tcp{$i$ is initiator, $r$ is responder}
        \Switch{$i.\Leader$, $r.\Leader$}{
            \Case{$1,1$}{
                \tcp{responder loses \Leader bit}
                $r.\Leader ← 0$
            }
            \Case{$1,0$}{
                \tcp{\Leader bit moves to responder}
                $i.\Leader ← 0$\;
                $r.\Leader ← 1$
            }
            \Case{$0,1$}{
                \tcp{responder changes $r.y$ if not frozen}
                \If{$r.\Frozen = 0$}{
                    $r.y ← (r.y + 1) \bmod \card{Y}$
                }
            }
            \Case{$0,0$}{
                \tcp{perform step of $R(\Tuple{x,y})$}
                $i.q, r.q ← δ_R(i.q, r.q)$
            }
        }
    }
    \caption{Stably computing $R(x,y)$ from $R(\Tuple{x,y})$}
    \label{alg-predicate-implies-relation}
\end{algorithm}

In discussing Algorithm~\ref{alg-predicate-implies-relation}, it will
be convenient to abbreviate the four cases by their values; for
example, a transition that reduces the number of leaders will be a
$1,1$ transition.

To prove the correctness of
Algorithm~\ref{alg-predicate-implies-relation}, we start with a simple
invariant, the proof of which is immediate from inspection of the
code:
\begin{lemma}
    \label{lemma-leader-invariant}
    Every reachable configuration of
    Algorithm~\ref{alg-predicate-implies-relation} has at least one
    agent with a nonzero \Leader bit.
\end{lemma}

Next, we show that every output permitted by $R(\Tuple{x,y})$ is reached in some execution:
\begin{lemma}
    \label{lemma-stabilize-on-y}
    If $R(\Tuple{x,y})$ holds and $n≥3$, then there exists an execution of
    Algorithm~\ref{alg-predicate-implies-relation} on a complete
    interaction graph with input $x$ that
    stabilizes with output $y$.
\end{lemma}
\begin{proof}
    Construct the execution in stages:
    \begin{enumerate}
        \item Apply $1,1$ transitions until there is a single nonzero \Leader
            bit.
        \item Let $\ell$ be the agent with the nonzero \Leader bit. Apply $0,1$
            transitions until $\ell.y = y_\ell$. Apply a $1,0$
            transition to move $\ell$ to the next agent and repeat
            until $i.y = y_i$ for all agents $i$.
        \item Finally, run the embedded $R(\Tuple{x,y})$ protocol
            until it stabilizes with $i.\Frozen = 1$ for all agents
            $i$. This
            requires a bit of care to deal with the \Leader bits: if
            the embedded execution requires an interaction between $i$
            and $r$ with $i.\Leader = r.\Leader = 0$, no special treatment is
            required, but in the event that $i.\Leader=1$ or $r.\Leader=1$, use a
            $1,0$ transition to move the \Leader bit to a third agent
            before having $i$ and $r$ interact.
    \end{enumerate}

    At the end of these steps, the protocol will be in a configuration
    with output $y$, and this configuration will be output-stable
    because $i.\Frozen = 1$ for all agents, preventing any changes
    to the input to the embedded $R(\Tuple{x,y})$ protocol or to the output to
    the protocol as a whole.
\end{proof}

We also need to show that we do not stabilize on incorrect outputs:
\begin{lemma}
    \label{lemma-stable-outputs-are-good}
    If Algorithm~\ref{alg-predicate-implies-relation} stabilizes on
    output $y$ on a complete interaction graph when $n≥3$, then $R(\Tuple{x,y})$ holds.
\end{lemma}
\begin{proof}
    Let $C$ be an output-stable configuration of
    Algorithm~\ref{alg-predicate-implies-relation}.
    Then $i.\Frozen = 1$ for all $i$ in all configurations reachable
    from $C$, as otherwise there is a $0,1$ transition that changes
    the output.

    With $n≥3$, the embedded $R(\Tuple{x,y})$ execution
    satisfies global fairness, since any transition it needs to do can
    be arranged by moving the leader bit out of the way as in the
    proof of Lemma~\ref{lemma-stabilize-on-y}.
    Furthermore the inputs to the embedded protocol are now stable.
    It follows that the embedded protocol eventually stabilizes to the
    correct value of $R(\Tuple{x,y})$; but we have already established
    that this value is $1$.
\end{proof}

Finally, we show that Algorithm~\ref{alg-predicate-implies-relation}
stabilizes in all fair executions:
\begin{lemma}
    \label{lemma-always-stabilize}
    Let $R(\Tuple{x,y})$ be a stably computable predicate such that
    $R(x,y)$ is total.
    Then in any globally fair execution on a complete interaction
    graph with $n≥3$,
    Algorithm~\ref{alg-predicate-implies-relation} is eventually
    output-stable.
\end{lemma}
\begin{proof}
    It is enough to show that from any reachable configuration $C$, there
    is an output-stable configuration $C'$ with $C\stackrel{*}{→}C'$.
    This is because in any infinite execution, there is at least one
    reachable $C$ that occurs infinitely often. Iterating global fairness then
    implies that the output-stable $C'$ also occurs infinitely often,
    which implies that it occurs at least once.

    Fix some reachable configuration $C$. Construct a path from $C$ to
    an output-stable $C'$ following essentially the same formula as in
    the proof of Lemma~\ref{lemma-always-stabilize}:
    \begin{enumerate}
        \item Use $1,1$ transitions to reduce the number of leaders to exactly one if needed.
        \item Let $x$ be the input and let $y(C)$ be the output in
            $C$. If $R(\Tuple{x,y(C)})$ does not hold, run the embedded
            $R(\Tuple{x,y})$ tester until $i.\Frozen = 0$ for all $i$,
            then apply $0,1$ and $1,0$ transitions to change the
            output to some $y$ for which $R(\Tuple{x,y})$ does hold.
            Such a $y$ exists since we are given that $R(x,y)$ is total.
        \item Run the embedded $R(\Tuple{x,y})$ tester until
            $i.\Frozen$ stabilizes to $1$ for all $i$.
    \end{enumerate}
\end{proof}

Combining the lemmas gives:
\begin{theorem}
    \label{theorem-predicate-implies-relation}
    Let $R(\Tuple{x,y})$ be a predicate that is stably computable for
    $n≥3$ on a complete interaction graph, such that
    $R(x,y)$ is total.
    Then $R(x,y)$ is a stably computable relation on for $n≥3$ on a
    complete interaction graph.
\end{theorem}
\begin{proof}
    Lemma~\ref{lemma-always-stabilize} shows that
    Algorithm~\ref{alg-predicate-implies-relation} is eventually
    output-stable in all fair executions with $n≥3$.
    Let $R'(x,y)$ be the relation stably computed by
    Algorithm~\ref{alg-predicate-implies-relation} in these executions.
    Lemma~\ref{lemma-stabilize-on-y} shows that $R(\Tuple{x,y})$
    implies $R'(x,y)$, and Lemma~\ref{lemma-stable-outputs-are-good}
    shows that $R'(x,y)$ implies $R(\Tuple{x,y})$. So $R(x,y) =
    R'(x,y)$ is stably computed by
    Algorithm~\ref{alg-predicate-implies-relation} for $n≥3$.
\end{proof}

\section{Not all stably computable relations give stably computable
predicates}

First let us show that reachability is a stably computable relation:
\begin{theorem}
    \label{theorem-reachability}
    Let $P$ be a population protocol on a complete interaction graph and let $C\stackrel{*}{→}C'$ if
    there is a sequence of zero or more steps of $P$ starting in
    configuration $C$ and ending in $C'$.
    Then the relation $\stackrel{*}{→}$ is stably computable on a
    complete interaction graph for
    $n≥3$.
\end{theorem}
\begin{proof}
    The idea is to simulate $P$ for a while and then
    nondeterministically choose to freeze its configuration.
    Let $δ$ be the transition function for $P$.
    The state space for the simulation algorithm is $\Set{a,b,-}×Q_P$,
    and the transition function is given by the following rules:
    \begin{displaymath}
        \begin{array}{ccccc}
            (a,q_1) & (a,q_2) & → & (a,q_1) & (b,q_2) \\
            (b,q_1) & (b,q_2) & → & (b,q_1) & (-,q_2) \\
            (a,q_1) & (-,q_2) & → & (-,q_1) & (a,q_2) \\
            (b,q_1) & (-,q_2) & → & (-,q_1) & (b,q_2) \\
            (-,q_1) & (a,q_2) & → & (a,q_1) & (-,q_2) \\
            (-,q_1) & (b,q_2) & → & (b,q_1) & (-,q_2) \\
            (a,q_1) & (b,q_2) & → & (a,δ_1(q_1,q_2)) & (b,δ_2(q_1,q_2) \\
            (b,q_1) & (a,q_2) & → & (-,q_1) & (-,q_2) \\
        \end{array}
    \end{displaymath}

    Initially all agents start with an $a$ token and a state $q$ drawn
    from the input configuration $C$. The first two rules
    eventually prune these down to a single $a$ and a single $b$
    token, leaving the remaining agents with a $-$ (blank) token.

    Any $ab$ interaction updates the $q$ components of each agent
    according to the transition relation for $P$.

    The $ba$ interaction removes one pair of tokens; once no $a$
    tokens are left or no $b$ tokens are left, the output is stable.
    This eventually occurs in all fair executions.

    While there are many possible ways an execution of the simulation
    can go, starting from a particular input $C$ with $n≥3$ we can
    stabilize on a particular $C'$ with $C\stackrel{*}{→}C'$ as follows:
    \begin{enumerate}
        \item Apply the $aa$ rule until there is exactly one $a$
            token.
        \item Apply the $bb$ rule until there is exactly one $b$
            token.
        \item For each step of an execution $C\stackrel{*}{→}C'$, let
            $i$ and $r$ be the initiating and responding agents. Use
            the $a-$ and $b-$ rules to move the unique $a$ and $b$
            tokens onto $i$ and $r$, then apply the $ab$ rule to
            update the simulated $P$ states of $i$ and $r$. This is
            the step that requires $n≥3$, since we may need to move
            the $a$ or $b$ token to a third agent if we need to swap
            them between $i$ and $r$.
        \item Once $C'$ is reached, use the $ba$ rule to remove the
            $a$ and $b$ tokens. Since all agents now carry a blank
            token, the $a,b$ rule is no longer enabled, and the $Q$
            components of the agents never change from $C'$.
    \end{enumerate}

    Conversely, since the $a,b$ transitions only update states
    according to $δ$, any output $C'$ generated by this protocol
    satisfies $C\stackrel{*}{→}C'$. So the protocol stably computes
    $\stackrel{*}{→}$ for $n≥3$.
\end{proof}

Next, we show that reachability for population protocols is not, in
general, semilinear. The idea is to adapt a classic non-semilinear construction of
Hopcroft and Pansiot~\cite{HopcroftP1979} for vector addition systems
to the specific case of population protocols.\footnote{This
construction is mentioned in the work of
Angluin~\etal~\cite{AngluinAER2007} showing semilinearity of stably
computable predicates, but in the more general context of arbitrary
vector addition systems, which requires no change to the original
construction of Hopcroft and Pansiot~\cite{HopcroftP1979}. While adapting the construction to the more
specific context of population protocols is not difficult, we do so
here for completeness.}

\begin{lemma}
    \label{lemma-not-semilinear}
    There exists a population protocol for which the reachability
    relation $\stackrel{*}{→}$ on a complete interaction graph is not semilinear.
\end{lemma}
\begin{proof}
    Consider the population protocol with state set and
    a transition function given by the following rules, with the usual
    convention that omitted rules are no-ops:
    \begin{align*}
        1c &→ 2d \\
        2b &→ 1d \\
        1a &→ 3b \\
        3d &→ 4c \\
        4b &→ 3c \\
        3a &→ 1b \\
    \end{align*}

    The intent is that the numbered states are held by a unique
    finite-state controller agent, while the remaining states are used
    to generate up to $2^k$ $c$ and $d$ tokens starting from $k$ $a$
    tokens, one $c$ token, and a large supply of $b$ (blank) tokens.

    Intuitively, the $1$ and $2$ states expand $c$ tokens into pairs of
    $d$ tokens, with the second $d$ token obtained by recruiting an
    agent with a $b$ token. The $3$ and $4$ states do the reverse,
    turning each $d$ into two $c$'s. Swapping the direction of
    doubling consumes an $a$ token, limiting the maximum number of doublings 
    during an execution to the number of initial $a$'s.

    We can prove this by constructing an appropriate invariant.
    Let us abuse notation by using $a,b,c,d$ for the count of the
    number of $a,b,c,d$ agents in a particular configuration. 
    For a configuration $C$ with exactly one numbered agent, let $q$ be
    the state of that agent, and
    define the
    function
    \begin{align*}
        f(C) &=
        \begin{cases}
            2^a (2c+d) & \text{when $q=1$} \\
            2^a (2c+d+1) & \text{when $q=2$} \\
            2^a (2d+c) & \text{when $q=3$} \\
            2^a (2d+c+1) & \text{when $q=4$} \\
        \end{cases}
    \end{align*}

    It is straightforward to check that $f$ is non-decreasing for any
    transition and preserved by at least one possible non-trivial
    transition if non-trivial transitions are available. We use the usual
    convention that any values labeled with a prime are the new values 
    following the transition.
    \begin{itemize}
        \item 
            A $1c$ transition sets $c' = c-1$, $d'
            = d+1$, and $q=2$, leaving $f' = 2^a(2c' + d' + 1) = 2^a(2(c-1) +
            d + 2) = 2^a(2c+d) = f$
        \item 
            A $2b$ transition sets $d'
            = d+1$, and $q=1$, leaving $f' = 2^a(2c + d') =
            2^a(2c+d+1) = f$.
        \item
            A $1a$ transition sets $a' = a-1$ and $q=3$, leaving
            $f' = 2^{a-1} (2d+c) ≤ 2^a(2c+d) = f$, with equality when
            $c=0$.
        \item The remaining transitions are symmetric.
    \end{itemize}

    Consider now an execution that starts with $q=1$, $a=k$, $c=1$, and
    $d=0$. This gives $f = 2^a (2c + d) = 2^{k+1}$. With sufficiently
    many initial $b$ tokens, it is possible to apply only transitions
    at each step that preserve equality, eventually converging to a
    configuration with $a=0$ and one of $c$ or $d$ also equal to $0$.
    This will make whichever of $c$ or $d$ is not zero equal to
    $2^{k+1}$. Furthermore, the sum $c+d$ is non-decreasing throughout
    this execution, meaning that any value $c+d ≤ 2^{k+1}$ appears in
    some reachable configuration.

    Suppose now that reachability for this protocol is semilinear. 
    Using Presburger's characterization of semilinearity predicates as
    precisely those definable in first-order Presburger
    arithmetic~\cite{Presburger1929}, this would imply the existence
    of a first-order Presburger formula $φ$ such that $φ(C,C')$ holds
    if and only if $C \stackrel{*}{→} C'$, where
    the parameters to $φ$ represent counts of agents in particular
    states. Writing $x_{C}$ for the number of agents in state $x$ in
    configuration $C$, we can define a new first-order Presburger
    formula in variables $\ell$
    and $k$ by
    \begin{align*}
        ∃C,C':
         φ(C,C')
        &∧ (1_C = 1) 
        ∧ (2_C = 0) 
        ∧ (3_C = 0) 
        ∧ (4_C = 0) 
        \\&∧ (a_C = k)
        ∧ (c_C = 1)
        ∧ (d_C = 0)
        \\&∧ (c_{C'} + d_{C'} = \ell).
    \end{align*}
    This sets the initial configuration $C$ to the
    initial configuration in a class of executions yielding $1 ≤ c+d ≤
    2^{k+1}$ throughout and sets $\ell = c+d$. But the set
    $\SetWhere{(\ell,k)}{1≤\ell ≤ 2^{k+1}}$ is not semilinear,
    contradicting the assumption that $\stackrel{*}{→}$ is semilinear.
\end{proof}

Since all stably computable predicates are
semilinear~\cite{AngluinAER2007}, it follows that:
\begin{corollary}
    \label{corollary-stably-computable-relation-not-predicate}
    There exists a stably computable relation $R(x,y)$ such that the
    corresponding predicate $R(\Tuple{x,y})$ is not stably computable.
\end{corollary}

\section{Single-valued stably computable relations give stably computable predicates}

\begin{theorem}
    \label{theorem-single-valued}
    Let $R(x,y)$ be a single-valued relation. Then if $R(x,y)$ is
    stably computable, so is $R(\Tuple{x,y})$.
\end{theorem}
\begin{proof}
    Let $P$ be a protocol that stably computes $R(x,y)$. Since $R$ is
    single-valued, the output of $P$ on input $x$ always converges to
    the same value $y$. To stably compute $R(\Tuple{x,y'})$, construct
    a protocol that uses $P$ to compute the unique $y$ such that
    $R(x,y)$ holds, then make this output the input to a
    test for $y=y'$ using the stabilizing-inputs construction
    of~\cite{AngluinACFJP2005}.
\end{proof}

This construction only works for single-valued relations. If some $x$
has distinct $y_1≠y_2$ with $R(x,y_1)$ and $R(x,y_2)$, then if the
embedded computation stably outputs $y_1$, the protocol will fail to recognize
$R(\Tuple{x,y_2})$. From
Corollary~\ref{corollary-stably-computable-relation-not-predicate} we know that 
there can be no general method to overcome this problem, although for
specific multiple-valued $R(x,y)$ it may be that $R(x,y)$ is
stably computed by some specialized protocol.

\section{Handling small cases}
\label{section-small-cases}

\newData{\Rank}{rank}
\newData{\Index}{index}
\newData{\Data}{data}

The assumption that $n$ is sufficiently large can be removed subject
to a small technical condition on executions involving only one
agent.
Because this construction may be useful in other
contexts, we state it as a generic lemma:
\begin{lemma}
    \label{lemma-small-cases}
    Let $R(x,y)$ be a relation that is stably computable for $n≥n_0$
    and single-valued for $n=1$.
    Then $R(x,y)$ is stably computable for all $n$.
\end{lemma}
\begin{proof}
    Let $P$ be a protocol that stable computes $R(x,y)$ for $n≥n_0$.
    We will build a wrapper protocol that stably computes $R(x,y)$ for
    $n<n_0$ and switches its output to the output of $P$ for larger
    $n$.

    The wrapper protocol works by assigning each agent a value
    $\Rank$ in the range $0$ through $n_0$, using a variation on the
    self-stabilizing leader election protocol of
    Cai~\etal~\cite{CaiIW2012}.\footnote{The name of the $\Rank$ value
    follows the use of this protocol by
    Burman~\etal~\cite{BurmanCCDNSX2021} in the context of
    the more general problem of assigning agents unique ranks.}
    The $\Rank$ protocol starts all processes with $\Rank$ $0$ and 
    consists of a single rule $k,k→k,\min(k+1,n_0)$. 
    This eventually assigns each process a stable rank, which will be
    a unique value in the range
    $0\dots n-1$ when $n<n_0$. When $n≥n_0$, at least one process will
    be assigned $n_0$.

    Nondeterminism is provided by an $\Index$ variable stored at the
    (eventually unique) agent with $\Rank$ value $0$, which is
    eventually locked by setting a $\Frozen$ bit that is initially $0$
    for all agents. Let $m$ be the number of distinct possible outputs
    $y$ for populations of size less than $n_0$. Each time an agent
    with $\Rank = 0$ and $\Frozen = 0$ is an initiator, it increments
    $\Index$ mod $m$. If an agent with $\Rank=0$ is a responder, it
    sets $\Frozen = 1$. This ensures that $\Index$ is eventually
    stable while still allowing $\Index$ to take on all possible
    values in the range $0 \dots m-1$ provided $n≥2$.

    Finally, each agent records in an array location $\Data[r]$,
    initially $⊥$,  the
    most recent $\Index$ and input it has observed in any agent with
    $\Rank=r$. Since inputs are fixed and $\Rank$ and $\Index$
    eventually stabilize for all agents, this causes all agents to
    eventually obtain identical $\Data$ arrays when $n<n_0$. When
    $n≥n_0$, $\Data[n_0]$ does not necessarily get the same value at
    all agents, but it does eventually get a non-$⊥$ value at all
    agents.

    Each step of the wrapper protocol is carried out in parallel with
    a step of protocol $P$. No interaction occurs between these
    protocols except in the output function $O$.

    The $\Data$ array and the output of $P$ are used as arguments to 
    $O$. When $\Data[n_0] = ⊥$, each agent can (a) compute $n$ based
    on the largest non-$⊥$ entry in $\Data$, (b) compute the input $x$
    to the protocol from the inputs stored in $\Data$, and (c)
    observe the same $\Data[0].\Index$ value as the other agents.
    Let $S = \SetWhere{y}{R(x,y)}$. The agent with rank $r$ selects the $k$-th
    element $y$ of $S$ where $k = \Data[0].\Index \bmod
    \card{S}$, then outputs $y_r$; in this way the agents collectively
    output the $k$-th value $y∈S$. When $\Data[n_0] ≠ ⊥$, the agent
    outputs whatever its output is in the embedded $P$ computation.

    We now argue that this combined protocol stably computes $R(x,y)$.
    For $n≥n_0$, this follows from the fact that the combined protocol
    stabilizes on the output of $P$, which stably computes $R(x,y)$
    for $n≥n_0$. For $1<n<n_0$, with input $x$, every possible
    output $y$ satisfies $R(x,y)$, and all such $y$ can be obtained by
    an appropriate setting of the common value $\Data[0].\Index$. 
    The remaining case $n=1$ is covered by the assumption that
    $R(x,y)$ is single-valued when $n=1$; the unique agent never
    updates its state and outputs the unique $y$ consistent with its
    input $x$.
\end{proof}

The technical requirement that $R(x,y)$ is single-valued when $n=1$ is
necessary because a population protocol with only one agent cannot
take any steps, making its output a deterministic function of its input. An
alternative might be to insist that any population protocol have at
least two agents, but regrettably this very natural assumption does
not appear in the original definition of the model by
Angluin~\etal~\cite{AngluinADFP2006}.

\section{Conclusions}

We have considered the connection between relations
$R(x,y)$ where $x$ is an input to a population protocol and $y$ a
possible stable output, and their corresponding predicates
$R(\Tuple{x,y})$, where both $x$ and $y$ are provided as inputs.
We have shown that, for population protocols
running on a complete interaction graph with $n≥2$,
total stably computable predicates always give stably computable
relations, but that the converse is not necessarily true unless $R$ is
single-valued.

The proof of some of these results depends on the identity between
stably computable predicates on a complete interaction graph and semilinear
predicates. A natural question is whether these results generalize to
other classes of interaction graphs for which
this property does not hold. Results based on the nondeterministic
freezing technique used in
Theorem~\ref{theorem-predicate-implies-relation} 
and
Theorem~\ref{theorem-reachability} 
seem like good candidates for generalization, but it is not
immediately obvious how to adapt the current handling of
nondeterminism in these constructions to
more general graphs. We
leave exploration of this question to future work.

\bibliographystyle{alpha}
\bibliography{paper.bib}

\end{document}